\documentclass[11pt]{article}

\usepackage{fullpage}
\usepackage[utf8]{inputenc}
\usepackage{graphicx}
\usepackage{amsmath,amsthm,amssymb}
\usepackage{algorithm,algpseudocode}
\usepackage{enumerate}
\usepackage{color}
\usepackage[none]{hyphenat}
\usepackage{xspace}
\usepackage{caption,subcaption}
\usepackage{nccmath}

\newtheorem{theorem}{Theorem}
\newtheorem{lemma}[theorem]{Lemma}

\newtheorem{claim}{Claim}

\algdef{SE}[SUBALG]{Indent}{EndIndent}{}{\algorithmicend\ }%
\algtext*{Indent}
\algtext*{EndIndent}

\newcommand{\OPT}{\emph{OPT}\xspace}
\newcommand{\MSPEC}{\emph{MSPEC}\xspace}
\newcommand{\MSS}{\emph{MSS}\xspace}

\newcommand{\UDG}{\emph{UDG}}

\title{Minimum Shared-Power Edge Cut \footnote{This work was partially supported by NSERC and by the Slovenian Research Agency (P1-0297).}}

\author{
Sergio Cabello\thanks{University of Ljubljana and IMFM, \texttt{sergio.cabello@fmf.uni-lj.si}.}
\and
Kshitij Jain\thanks{University of Waterloo, \texttt{k22jain@uwaterloo.ca}.}
\and
Anna Lubiw \thanks{University of Waterloo, \texttt{alubiw@uwaterloo.ca}.}
\and
Debajyoti Mondal \thanks{University of Saskatchewan, \texttt{dmondal@cs.usask.ca}.}
}

\begin{document}

\clearpage
\maketitle
\thispagestyle{empty}

\begin{abstract}

We introduce a problem called the Minimum Shared-Power Edge Cut (\MSPEC).  The input to the problem is an undirected edge-weighted graph with distinguished vertices $s$ and $t$, and the goal is to find an $s$-$t$ cut by assigning ``powers'' at the vertices and removing an edge if the sum of the powers at its endpoints is at least its weight. The objective is to minimize the sum of the assigned powers.

\MSPEC is a graph generalization of a barrier coverage problem in a wireless sensor network: given a set of unit disks with centers in a rectangle, 
what is the  minimum total amount by which we must shrink the disks to permit
an intruder to cross the rectangle undetected, i.e.~without entering any disc.
This is a more sophisticated measure of barrier coverage than the minimum number of disks whose removal breaks the barrier.

We develop a fully polynomial time approximation scheme (FPTAS) for \MSPEC. We give polynomial time algorithms for the special cases where the edge weights are uniform, or the power values are restricted to a bounded set. Although \MSPEC is related to network flow and matching problems, its computational complexity (in P or NP-hard) remains open. 
\end{abstract}

\newpage
\setcounter{page}{1}

\section{Introduction}

Minimum weight edge cuts in graphs are very well-studied. In this paper we look at a variation that arises from unit disk graphs and other situations where the edges of the graph are determined by geometric properties of the vertices. In this variation, we assign a \emph{power} $p_v$ to each vertex $v$ and an edge $e=(u,v)$ is removed if $p_u + p_v$ is at least the weight of edge $e$.  The goal is to remove the edges of an $s$-$t$ cut while minimizing the power sum.  More formally, the Minimum Shared-Power Edge Cut Problem (\MSPEC) is defined as follows:

\smallskip
{\bf Input:} Graph $G=(V \cup \{s,t\},E)$ with $n$ vertices, $m$ edges, and a non-negative weight $w_{u,v}$ on each edge $(u,v) \in E$.

{\bf Problem:} Assign a non-negative power $p_v$ to each $v \in V$ and assign $p_s=p_t=0$ so that removing the edge set $\{(u,v) \in E : p_u + p_v \ge w_{u,v} \}$ disconnects 
$s$ and $t$ and $\sum_{v \in V} p_v$ is minimized.
\smallskip

Our main result is a fully polynomial time approximation scheme (FPTAS) for the Minimum Shared-Power Edge Cut Problem.  

A special case of \MSPEC---and our original motivation for studying it---is the problem of measuring barrier coverage of a sensor network.  A sensor network is typically modelled as a collection of unit disks in the plane, where each disk represents the sensing region of its corresponding sensor~\cite{Wang:2011:CPS:1978802.1978811}. 
The network provides a \emph{barrier} between regions $R_1$ and $R_2$ if every path from $R_1$ to $R_2$ intersects the union of the disks. 
One simple measure of barrier coverage is the minimum number of disks whose removal permits a path from $R_1$ to $R_2$ in the free space outside the disks. 
This can be computed in polynomial time~\cite{kumar2007barrier} for the special case of a \emph{rectangular barrier}, where the sensors lie in a rectangle and must block paths from the bottom to the top of the rectangle.
We suggest, as a more sophisticated measure, the minimum total amount by which we must shrink the disks to permit such a path. 
This measure reflects the reality that sensor strength typically deteriorates (``attenuates'') with distance from the sensor~\cite{Wang:2011:CPS:1978802.1978811}. 
For a rectangular barrier, our Minimum Shrinkage problem (as in Figure~\ref{fig:barrier}) can be modelled as \MSPEC, where we have a vertex for each disk, and the power assigned to a vertex  
tells us how much to shrink the corresponding disk.
Thus our main result provides an approximation scheme to compute Minimum Shrinkage for rectangular barrier coverage.

\begin{figure}
\captionsetup[subfigure]{justification=centering,labelfont=normal}

    \centering
    \begin{subfigure}[b]{0.4\textwidth}
            \centering
            \includegraphics[width=\textwidth]{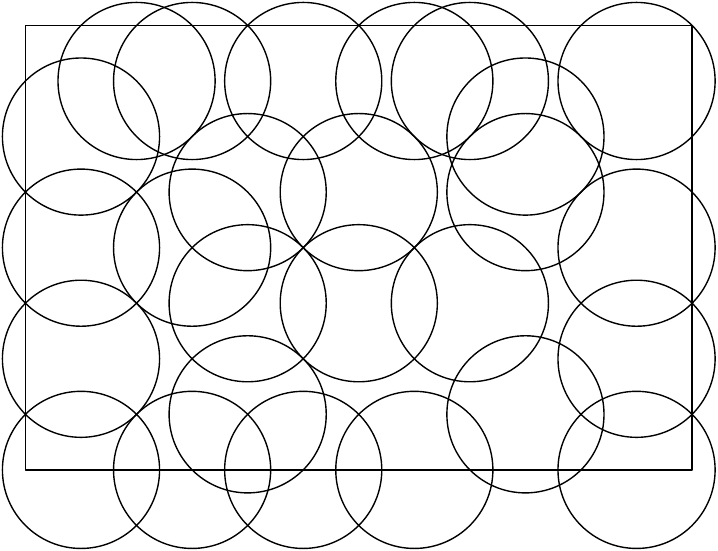}
    \caption{\label{fig:disks}}
    
    \end{subfigure}
\begin{subfigure}[b]{0.4\textwidth}
            \centering
            \includegraphics[width=\textwidth]{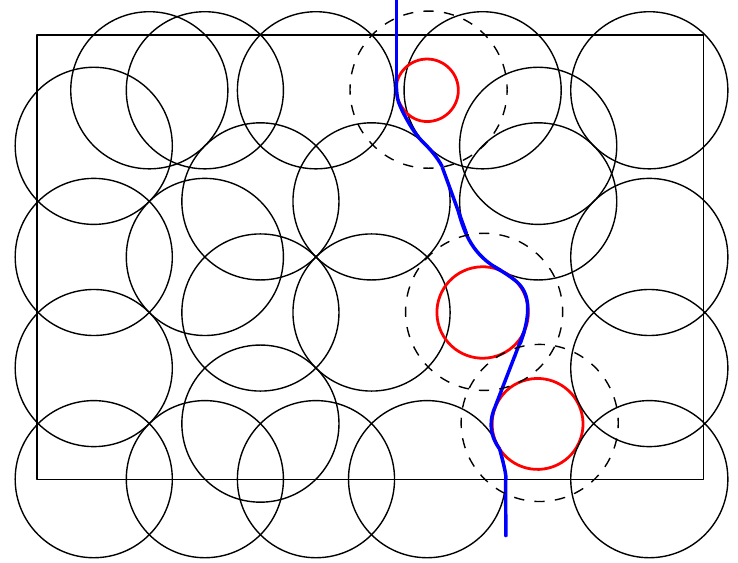}
    \caption{ \label{fig:shunk-disks}}
   
    \end{subfigure}
    \caption{Minimum shrinkage for the rectangular barrier coverage problem:  (\subref{fig:disks}) shows the rectangular barrier and the disks representing sensors;  (\subref{fig:shunk-disks}) shows a path from bottom to top in the free space after shrinking three of the disks. 
    }
    \label{fig:barrier}
\end{figure}

\subsection{Related Models and Work}

Motivated by problems in design of wireless networks, there
is substantial previous work on graph optimization problems in which edges are selected depending on ``power'' values $p_v$ that are assigned to each vertex $v$.  The objective is to minimize the sum of the powers, while satisfying  
connectivity properties for the selected edges that are of interest in the design of ad-hoc wireless networks. For example, 
the selected edges form a spanning tree, or a $k$-connected subgraph.  
The models that have been considered for selecting, or ``activating'', an edge $(u,v)$ of weight $w_{u,v}$  are: 
\begin{enumerate}
    \item $\min\{p_u, p_v\} \ge w_{u,v}$ (``Power Optimization'', e.g.~\cite{hajiaghayi2007power})
    \item $\max\{p_u, p_v\} \ge w_{u,v}$ (\cite{angel2015min}) 
    \item $p_u + p_v \ge w_{u,v}$ (``Installation Cost''~\cite{Panigrahi:2011:SND:2133036.2133114}) 
    \item a more general function of $p_u$ and $p_v$ (``Activation Networks''~\cite{Panigrahi:2011:SND:2133036.2133114})
\end{enumerate}

The only previous work in Activation Networks on our model (3) is an $O(\log n)$ approximation algorithm for minimum spanning tree, see~\cite{Panigrahi:2011:SND:2133036.2133114}.  
The only Activation result for the $s$-$t$ cut problem in any of these models is due to Angel et al.~\cite{angel2015min} who use model (2) and reduce the problem to a conventional min cost $s$-$t$ cut problem, which can then be solved in polynomial time.  They heavily use the fact that in model (2) the power assigned to a vertex will be equal to the weight of an incident edge.  This property does not hold for our model, nor do we assume (as is done in model (4)) a discrete set of power values. 
In Section \ref{subsec:polynomial domain}
we show that with that assumption, our problem can be solved in polynomial time.

In the most well-studied model, called ``Power Optimization,'' an edge $(u,v)$ with weight $w_{u,v}$ is selected, or ``activated'', if $p_u$ and $p_v$ 
are both at least $w_{u,v}$, i.e., $\min\{p_u, p_v\} \ge w_{u,v}$.  
This makes the problems easier because the power at a vertex will be equal to the 
weight of one of its incident edges, so the possible powers form a discrete set.
The literature on Power Optimization includes approximation algorithms for minimum spanning tree, Steiner forest, $k$-vertex or $k$-edge connected subgraph (all-pairs, single source, or $s$-$t$), and various degree-constrained and edge-cover problems~\cite{althaus2006power,hajiaghayi2007power,lando2010minimum,cohen2015approximating}.

In a less well-studied scenario,  Angel et al.~\cite{angel2015min} defined an edge $(u,v)$ to be activated if $\max \{p_u,p_v\}$ $ \ge w_{u,v}$.  
Again, the power at a vertex will be equal to the 
weight of one of its incident edges, so there are only a discrete set of possible powers that may be assigned to a node. 
Using this property, Angel et al.~reduced the problem of finding an optimum $s$-$t$ cut in this model to the conventional minimum cut 
problem.  
For their model, they also gave a polynomial time algorithm for optimum $s$-$t$ path, a 2-approximation for vertex cover, and a proof that spanning tree is hard to approximate. 

Our model, where an edge $(u,v)$ with weight $w_{u,v}$ is selected if $p_u + p_v \ge w_{u,v}$, was called the ``installation cost'' model by Panigrahi~\cite{Panigrahi:2011:SND:2133036.2133114} who reported an $O(\log n)$ approximation algorithm for minimum spanning tree~\cite{panigrahi2008minimum}.
He proved that this approximation factor is best possible assuming P $\ne$ NP~\cite{Panigrahi:2011:SND:2133036.2133114}.
We are unaware of any other work specifically on this model.

Panigrahi~\cite{Panigrahi:2011:SND:2133036.2133114} formulated a generalization of all these power requirements called ``Activation Networks,'' where there is a general 0-1 function $f$ on pairs of vertex powers and an edge $(u,v)$ is activated if $f(p_u,p_v) = 1$. In fact, his model is more general in that it allows for different functions $f$ for different pairs of vertices.  However, the model has a significant restriction in that the power values are constrained to a discrete set $D$ whose size is assumed to to be polynomially bounded in $n$.  
Panigrahi gave $O(\log n)$ approximation algorithms  
for various network survivability problems in this setting, such as minimum spanning tree, minimum Steiner forest, and $2$-edge/vertex connectivity.
These algorithms were improved in several follow-up papers~\cite{nutov2013survivable,alqahtani2013approximation}.

\medskip
\noindent\textbf{Organization.}
The remainder of our paper is organized as follows.  Section~\ref{sec:formulations} contains some alternative formulations of \MSPEC. The barrier coverage application is treated in Section~\ref{sec:MSS}.  Our approximation scheme for \MSPEC is in Section~\ref{sec:FPTAS}.  This is preceded by Section~\ref{sec:bottleneck} on a bottleneck version of the problem, and followed by Section~\ref{sec:speedup} on speeding up the approximation algorithm.  Some more general and some restricted versions of \MSPEC are considered in Section~\ref{sec:variations}, and Section~\ref{sec:conclusions} concludes the paper.

\section{Alternative Formulations of \MSPEC}
\label{sec:formulations}

In this section we give two alternative formulations of the \MSPEC problem.

In \MSPEC we are searching for an edge cut.  The set of edges in an $s$-$t$ cut form a bipartite graph.  Given an edge cut, the power assignment that will remove those edges corresponds to a minimum ``$w$-vertex cover'' \cite[chapter 17]{schrijver2002combinatorial}, that is, an assignment of weights $y_v$ to the vertices so that $y_u + y_v \ge w_{u,v}$ for all edges $(u,v)$ in the cut. 
In a bipartite graph the minimum $w$-vertex cover is dual to the maximum weight matching. Therefore, \MSPEC  can be alternatively stated as: given an edge-weighted graph, partition the vertices into two sets with $s$ in one set and $t$ in the other, and minimize the weight of a maximum matching of the edges crossing between the two sets, e.g, see Figure~\ref{fig:gap}(a).

We can also formulate \MSPEC as an integer program (ILP). Although we do not use it in our algorithms, it may lead to future developments. Let $\Pi_{st}$ be the set of all paths in $G$ from $s$ to $t$. We will have a non-negative variable $p_u$ for each $u \in V $ and a $0$-$1$ variable $x_{u,v}$ for each edge $(u,v) \in E$ with the intended interpretation that the cut edges have $x_{u,v}=1$.
\begin{ceqn}
\begin{equation}
\label{eqn:ILP}
	\begin{array}{ll@{\hspace{2mm}}ll}
			& \min \displaystyle \sum_{u \in V} p_u & \\
			& \textit{s.t.} \displaystyle \sum_{(u,v) \in \pi} x_{u,v} \geq 1 & \forall \pi \in \Pi_{st} \\
			& p_u + p_v \geq w_{u,v} x_{u,v} &\forall (u,v) \in E \\
			& x_{u,v} \in \{0,1\} & \forall (u,v) \in E \\
			& p_u \geq 0 \hspace{1mm} & \forall u \in V \\
			& p_s = p_t = 0
	\end{array}
\end{equation}
\end{ceqn}

\begin{figure}
	\centering
	\includegraphics[width=\textwidth]{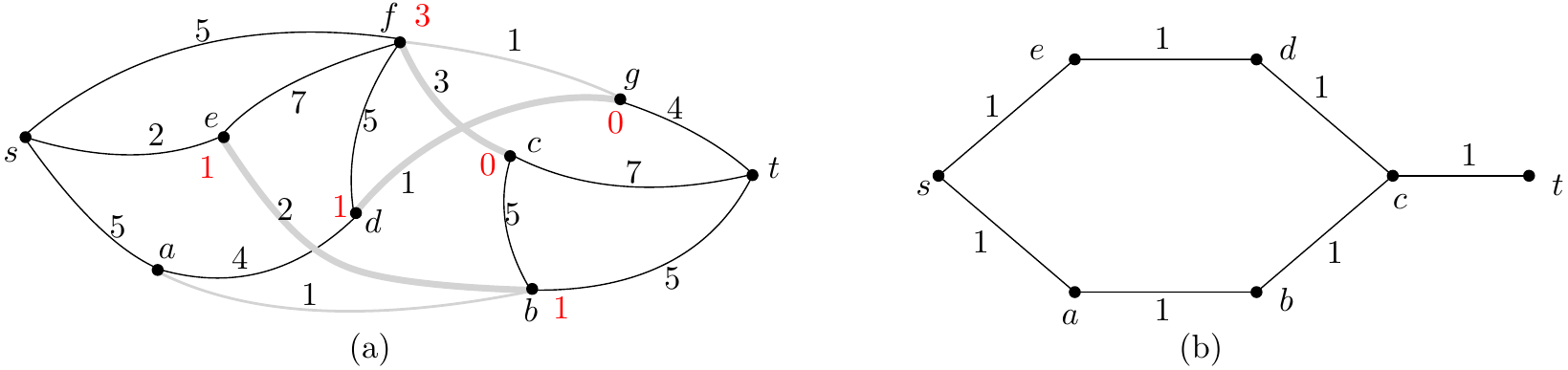}
	\caption{(a) Illustration for \MSPEC.  The powers on the vertices are shown in red.
	The cut edges are in gray and the edges of the maximum matching are shown in bold gray.
	(b) The integrality gap of ILP~(\ref{eqn:ILP}) is 2 for this graph.}
	\label{fig:gap}
\end{figure}

This integer programming formulation has an integrality gap of at least $2$. Consider the graph shown in Figure \ref{fig:gap}(b).  
It is easy to verify that the optimum integral solution has cost $1$. On the other hand, if we assign $p_c = 0.5$ and $x_{b,c}=0.5, x_{c,d}=0.5$ and $x_{c,t}=0.5$, this is a feasible solution for the LP relaxation of the above ILP. 
The integrality gap of this ILP is thus $ \geq \frac{1}{0.5}=2$.


\section{Barrier Coverage in a Sensor Network}
\label{sec:MSS}

Our motivation for studying the Minimum Shared-Power Edge Cut problem comes from a problem of measuring barrier coverage in a sensor network.
A sensor network consists of a set of sensors, where
each sensor is modelled as a unit disk (a disk of radius 1) in the plane.
The location of a sensor is the center point of its unit disk, and the sensor  
detects points within its unit disk. 

A sensor network provides a barrier between one region of the plane, $R_1$, and another region, $R_2$, if 
any point that travels from $R_1$ to $R_2$ will be detected by the sensors, i.e.,~there is no path between the two regions that remains outside all the disks. 

In a \emph{rectangular barrier} the sensors are located at points in a rectangle and the sensor network must detect any path that crosses the rectangle from below to above.  
Kumar et al.~\cite{kumar2007barrier} introduced this concept and suggested measuring ``barrier coverage'' as the minimum value $k$ such that every path that crosses the barrier intersects at least $k$ sensor disks.  Equivalently, the barrier coverage is the minimum number of sensors that must be removed to allow a path in the free space between disks.
For a rectangular barrier, Kumar et al.~showed that barrier coverage can be computed in polynomial time by applying Menger's theorem 
to the intersection graph of the disks---barrier coverage becomes the minimum size of a vertex cut, which is equal to the maximum number of vertex disjoint paths that go from the left of the rectangle to the right of the rectangle.

A number of alternative measures of barrier coverage have been proposed in the literature with the (conflicting) goals of having a measure that models reality and that can be computed efficiently. 
We propose a new measure called Minimum Shrinkage.  Like some of the previous measures (see the following section for more background) 
it models the fact that a sensor's ability to detect points decreases with distance from the sensor.
Our approximation algorithm for Maximum Shared-Power Edge Cut can be applied to compute Minimum Shrinkage for rectangular barrier coverage.

Our new measure is defined in Section~\ref{sec:MS-definition}, and the reduction to \MSPEC is given in Section~\ref{sec:MS-to-MSPEC}.  
We begin with further background on barrier coverage in 
Section~\ref{sec:barrier-background}.  

\subsection{Background on Barrier Coverage}
\label{sec:barrier-background}

As mentioned above, Kumar et al.~\cite{kumar2007barrier} introduced the idea of measuring barrier coverage by the minimum number of sensors whose removal permits a path from region $R_1$ to region $R_2$ in the free space outside the sensor disks.  

Bereg and Kirkpatrick~\cite{bereg2009approximating} applied this measure more broadly and called it ``resilience''.  
They introduced a new measure of barrier coverage called \emph{thickness} which is defined to be the minimum, over all paths from $R_1$ to $R_2$, of the number of times the path enters a disk, counting repeats.  
Whereas resilience seems to be hard to compute except for the case of a rectangular barrier, thickness can easily be computed in polynomial time using weighted shortest path algorithms. Bereg and Kirkpatrick showed that thickness provides a constant-factor approximation to resilience. 
Korman et al.~\cite{KormanLSS13} provided fixed-parameter tractable algorithms (parameterized by the resilience) and a PTAS when each point of the plane is covered by a few disks.

The above measures implicitly assume that a sensor can detect points uniformly across its unit disk.  In reality, the power of a sensor decreases with distance, and a sensor can detect closer points more easily than farther points.  
In order to take distance from the sensors into account, we can think of diminishing the power of the sensors, or shrinking their unit disks.
Meguerdichian et al.~\cite{meguerdichian-2001} formulated this as a ``maximum breach path''---to 
find the maximum value $d$ such that there is a path from region $R_1$ to $R_2$ where every point of the path is at least distance $d$ from every sensor point.
They showed that a maximum breach path should travel on the Voronoi diagram of the sensor points, and thus can be computed in time $O(n \log n)$.
The maximum breach measure is equivalent to asking for the maximum value $d$ such that shrinking all the sensor disks to radius $d$ permits a path in the free space between the shrunken disks.

In a related paper Megerian et al.~\cite{megerian2002exposure} introduced the notion of ``exposure''.  This measure takes into account, not only how close the path goes to the sensor points, but also the amount of time that is spent close to sensor points.
More formally, the ``[all] sensor field intensity'' at a point is the sum of all the sensor's powers at that point, and ``exposure'' along a path is the integral of the sensor field intensity along the path. Computing a path of minimum exposure is a very  difficult continuous problem.   Djidjev~\cite{djidjev2010approximation} gave an approximation algorithm, using the idea of discretizing the domain.  
For the special case of a rectangular barrier, there is a min-max formula that relates the minimum exposure of a path from bottom to top of the rectangle and a maximum flow from left to right of the rectangle---see Strang~\cite{strang1983maximal} and Mitchell~\cite{mitchell1990maximum}.   

For further background refer to Wang's survey on coverage problems in sensor networks~\cite{Wang:2011:CPS:1978802.1978811} (in particular, Section 6).

\subsection{Minimum Shrinkage}
\label{sec:MS-definition}

We propose a new measure of barrier coverage, called \emph{minimum shrinkage}.
This new measure models the reality that a sensor's detection ability drops off with distance.   

To \emph{shrink} a unit disk by amount $s_i$ for $0 \le s_i \le 1$ means to decrease its radius by $s_i$, i.e.,~to replace the unit disk by a disk of radius $(1-s_i)$. 
Let $c_i$, $i = 1, \ldots, n$ be the locations of unit disk sensors in a sensor network that is supposed to act as a barrier 
between region $R_1$ and region $R_2$ in the plane.  
The \emph{minimum shrinkage} of the sensor network is the minimum $\sum s_i$ such that if we shrink the $i^{\rm th}$ sensor disk by $s_i$, $1 \le i \le n$, then the network no longer provides a barrier between regions $R_1$ and $R_2$, i.e., there is a path from $R_1$ to $R_2$ in the free space between the shrunken disks.
See Figure~\ref{fig:barrier}.

\subsection{Reduction from Minimum Shrinkage to \MSPEC}
\label{sec:MS-to-MSPEC}

We will show that the Minimum Shrinkage problem for rectangular barrier coverage can be formulated as an \MSPEC problem.

Given a set of $n$ points $P$ in $\mathbb{R}^2$, the \emph{unit disk graph} induced by $P$, denoted \UDG($P$), is an embedded graph with vertex set $P$ and an edge $(u,v)$ when the unit disks centered at $u$ and $v$ intersect, i.e., the edge set is $\{(u,v) : u,v \in P $ and dist$(u,v) \leq 2 \}$.

An instance of the Minimum Shrinkage problem for rectangular barrier coverage consists of 
a set of points $P$ inside a rectangle $R$.  To reduce to \MSPEC we start with the graph \UDG($P$).  Define the weight of edge $(u,v)$ to be $(2-\text{dist}(u,v))$. This is the amount of shrinkage at $u$ and $v$ that is needed to make the disks at $u$ and $v$ non-intersecting.
We add two special vertices $s$ and $t$.  Vertex $s$ is connected to every vertex whose unit disk intersects the left boundary of $R$, and vertex $t$ is connected to every vertex whose unit disk intersects the right boundary of $R$. 
To define the weight of these edges, let $R_l$ and $R_r$ be the lines through the left and right boundaries of $R$, respectively, 
and, for line $L$, let $\text{dist}(u,L)$ be the distance from point $u$ to line $L$.  
Define the weight of an edge $(s,u)$ to be $(1-\text{dist}(u, R_l))$.  This is the amount of shrinkage needed to make the disk at $u$ not intersect the left boundary of $R$.
Similarly, define the weight of an edge $(t,u)$ to be  $(1-\text{dist}(u, R_r))$.
It is straight-forward to show that \MSPEC on the resulting graph solves the \MSS problem where we interpret $p_v$ as the amount to shrink the disk centered at $v$.

Observe that this reduction still works (with obvious modifications) for the more general problem where each sensor disk has a specified radius, i.e.~the radii are not uniform.

\section{Minimum Bottleneck Shared-Power Edge Cut}
\label{sec:bottleneck}

To approximately solve the minimum shared-power edge cut problem, we will need a solution to another closely related problem which we refer to as the ``Minimum Bottleneck Shared-Power Edge Cut Problem". 
This problem is very similar to \MSPEC in that we want to assign powers to the vertices such that $s$ and $t$ become disconnected if we remove any edge where the sum of the powers on its endpoints is at least its weight.
The key difference is that instead of assigning different powers to every vertex we will assign the same power to each vertex in $V$ and minimize this ``bottleneck'' power. 
More precisely, we want the minimum value $p$ such that $s$ and $t$ become disconnected if we remove the edges $(u,v), u,v \in V$ with $2p \ge w_{u,v}$ and the edges $(u,v), u \in V, v \in \{s,t\}$ with $p \ge w_{u,v}$.

In the case of barrier coverage for a rectangular barrier, the minimum bottleneck shared-power edge cut is equivalent to the ``maximum breach path'' measure introduced by Meguerdichian et al.~\cite{meguerdichian-2001}.  Namely, we want to shrink all sensor disks by the same (minimum) amount to permit a path in the free space between disks. Meguerdichian et al.~computed the maximum breach path in polynomial time using Voronoi diagrams.  It is interesting that the problem can be solved in a more general non-geometric setting.

In the following we show how to obtain an exact solution to the minimum bottleneck shared-power edge cut in polynomial time. Let the set of 
removed edges in the optimal solution be $M^*$ and the optimal power be $p^*$.

\begin{claim}
There is an edge $e=(u,v) \in M^*$ which is \emph{tight} for $p^*$, i.e.,~$w_{u,v} = 2p^*$ if $u,v \in V$ or $w_{u,v} = p^*$ if $u \in V$ and $v \in \{s,t\}$.
	\label{claim:tight}
\end{claim}
\begin{proof}[Proof]
	Suppose $w_{u,v} < 2p^*$ for all edges in $M^*$ with both end points in $V$ and $w_{u,v} < p^*$ for all edges in $M^*$ with one end point in $V$ .  Then we can reduce $p^*$ to $(p^*-\varepsilon)$ 
	and still ensure that $w_{u,v} \leq 2(p^*-\varepsilon)$ for all edges in $M^*$ with both end points in $V$ and $w_{u,v} \leq p^*-\varepsilon$ for edges in $M^*$ with one end point in $V$, so $p^*$ was not the minimum bottleneck power.
\end{proof}

\begin{theorem}
    Minimum Bottleneck Shared-Power Edge Cut can be solved optimally in $O((n+m)\log m)$ time.
    \label{theorem:bottleneck}
\end{theorem}

\begin{proof}
Define the \emph{power requirement} of edge $e= (u,v)$ to be $\frac{1}{2}w_{u,v}$ if $u,v \in V$ or $w_{u,v} $ if $u \in V$ and $v \in \{s,t\}$.
By Claim \ref{claim:tight}, 
$p^*$ is equal to the power requirement for some edge $e \in E$, and we can use binary search on the power requirements to find the minimum $p$ such that 
$s$ and $t$ become disconnected if we remove the edges whose power requirement is at most $p$.
We sort all the power requirements in non-decreasing order
and perform binary search on this list.  
To test a value $p'$ from this list, we create a graph $G'$ with vertices of $G$ and edges of $G$ whose power requirement is greater than $p'$.   
 
If $s$ and $t$ are disconnected in $G^\prime$ we recursively perform a binary search on the power requirements between $0$ and $p^\prime$, otherwise we recurse on all higher power requirements. The running time for this approach is determined by the binary search, which takes time $O(\log m)$, and by checking connectivity of $s$ and $t$ in $G^\prime$, which takes time $O(n+m)$.  Therefore, the total running time is $O((n+m)\log m)$.
Here $n$ is the number of vertices and $m$ is the number of edges.
\end{proof}

\section{Approximation Scheme for \MSPEC}
\label{sec:FPTAS}

The idea of our approximation algorithm is to convert the \MSPEC problem to a minimum vertex cut problem.  Observe that if we can only assign power 0 or 1 to every vertex in $V$, then our problem is Minimum Vertex Cut---remove a minimum number of vertices to disconnect $s$ and $t$.  We will discretize our problem by replacing each vertex $v \in V$ by multiple copies of $v$ such that removing one copy corresponds to assigning a small fraction of the maximum power to $v$. 
A similar approach was used, in a geometric setting, by Agarwal et al.~\cite{agarwal2014union}.
We want to ensure that the discretization introduces an error of at most $\frac{\varepsilon}{n}$ for each vertex with respect to the optimum solution.

In order to carry out this plan, we need an upper bound on the maximum power we might assign to any vertex, and, for the error analysis, we need a lower bound on the optimum solution.  We obtain these bounds from the minimum bottleneck shared-power edge cut.  
Let $n = |V|$, let $p^*$ be the minimum power for the bottleneck shared-power edge cut problem, and let \OPT be the minimum power sum for \MSPEC.

\begin{lemma} $p^* \le \OPT \le n p^*$.  
\label{lemma:n-approx}
\end{lemma}
\begin{proof}
Assigning power $p^*$ to every vertex in $V$ provides a feasible solution to \MSPEC, and therefore $\OPT \le np^*$.

For the other inequality, let $p_{\max}$ be the maximum power assigned to any vertex of $V$ in an optimum solution to \MSPEC. Then $p_{\max} \le \OPT$.  Assigning $p_{\max}$ to every vertex in $V$ provides a solution to the minimum bottleneck shared-power edge cut.  Therefore $p^* \le p_{\max} \le \OPT$.  
\end{proof}

From this lemma, we know that the maximum power we might assign to a vertex is $np^*$.  
The lemma also implies that if we introduce an error of at most $\alpha = \frac{\varepsilon}{n} p^*$ for each vertex, then the total error over all vertices will be at most $\varepsilon p^* \le \varepsilon \OPT$.  Our plan is to construct a new graph $G' = (V',E')$ in which we replace each vertex of $V$ by $c$ copies, where each copy represents power $\alpha$.  Since the total power that we might assign to the vertex is $np^*$, the number of copies of the vertex that we need is $c = \lceil \frac{np^*}{\alpha} \rceil = \lceil \frac{n^2}{\varepsilon} \rceil$.  
We will replace each vertex of $V$ by a sequence of $c$ vertices, $v(0), v(1), \ldots, v(c-1)$.  

Removing the first $k$ vertices of the sequence will correspond to assigning power $k \alpha$ to $v$.  In $G'$ we will assign edges to the copies $v(i)$ to reflect this.  In particular, for $u,v \in V$ we put an edge $(u(i), v(j))$ in $G'$ if $i \alpha + j \alpha < w_{u,v}$. 
For $u \in V, x \in \{s,t\}$ we put an edge $(u(i),x)$ in $G'$ if $i\alpha < w_{u,x}$. A related idea of discretizing the choices using vertices combined with minimum cuts was developed by Hochbaum et al. \cite{hochbaum1993tight,HochbaumN94} for integer linear programms with at most two variables per inequality. They rely on an algorithm by Picard \cite{picard1976maximal} for finding the minimum-cost closure of a directed graph.

The precise construction of $G'$ is given in Algorithm 1. 

\begin{algorithm}
	\caption{Construction of $G'$}
	\label{algorithm:construction}
	\begin{algorithmic}
		
		\State $\forall v \in V$, make $c=\lceil \frac{n^2}{\varepsilon} \rceil$ copies of $v$ numbered $v(0),v(1),\dots,v(c-1)$ 
	    \State $V^\prime =  \{v(i) \mid \forall v \in V, 0 \leq i < c \} \cup \{s,t\}$
	    \smallskip
		\State $\alpha=\frac{\varepsilon p^*}{n}$
	    \smallskip
		\State $ E^\prime = \{ (u(i),v(j)) \mid i\alpha + j\alpha < w_{u,v} \} \cup \{ (u(i), x) \mid x \in \{s,t\}, i\alpha < w_{u,x} \}$	
	\end{algorithmic}
\end{algorithm}

Our approximation algorithm now proceeds as follows.  We find a minimum $s$-$t$ vertex cut in $G'$, denoted $K^*$,  and use it to define power values on the vertices of $G$ as given in Algorithm \ref{algorithm:approximation algorithm mspec}.

\begin{algorithm}
	\caption{Approximation Algorithm for \MSPEC}
	\label{algorithm:approximation algorithm mspec}
	\begin{algorithmic}
		\State Construct $G^\prime$ as in Algorithm~\ref{algorithm:construction}
		\State $K^* \gets $ minimum $s$-$t$ vertex cut for $G^\prime$
		\State $k^*_v = | \{ v(j) \in K^* \} | $ for $v \in V$
	    \State $p_{v} =  k^*_v \cdot \alpha$ for $v \in V$
		\State return $p=(p_{v})_{v \in V}$ as the solution for \MSPEC on $G$
	\end{algorithmic}
\end{algorithm}

To prove that this algorithm is an FPTAS we will first prove that the solution returned by the algorithm is indeed a feasible solution for \MSPEC.
Then, to bound the approximation factor, we will derive an upper bound on the solution returned by the algorithm in terms of the optimum solution to \MSPEC.   
We first prove some properties of $G'$ and of $K^*$.

\begin{claim}
    If $v \in V$ and $i_1 < i_2$, then in $G'$ the neighbourhoods of $v(i_1)$ and $v(i_2)$ are related by $N(v(i_1)) \supseteq N(v(i_2))$.
    \label{claim:neighbourhoods}
\end{claim}
\begin{proof}
We will show that $(v(i_2),x) \in E'$ implies $(v(i_1),x) \in E'$.

\noindent
Case 1.  $x = u(j)$ for $u \in V$. Since the edge $(v(i_2), u(j))$ is in $E'$ we have $i_2 \alpha + j \alpha < w_{v,u}$ so $i_1 \alpha + j \alpha <  w_{v,u}$ which implies  $(v(i_1),u(j)) \in E'$.

\noindent
Case 2.  $x \in \{s,t\}$.  Since the edge $(v(i_2), x)$ is in $E'$ we have $i_2 \alpha  < w_{v,x}$ so $i_1 \alpha  <  w_{v,x}$ which implies  $(v(i_1),x) \in E'$.
\end{proof}

\begin{claim} The copies of $v$ in $K^*$ are $v(0), v(1), \ldots, v(k^*_v -1)$.
	\label{claim:nestedstructure}
\end{claim}

\begin{proof}
We will prove that the copies of $v$ in $K^*$ form a prefix of $v(0), v(1), \ldots v(c-1)$.
Then the result follows since there are $k_v^*$ copies of $v$ in $K^*$.

Consider $i_1 < i_2$.  
By Claim~\ref{claim:neighbourhoods}, $N(v(i_1)) \supseteq N(v(i_2))$. 
     Now observe that if a graph contains vertices $u$ and $v$ with $N(u) \supseteq N(v)$ and $v$ is in a minimum vertex cut, then $u$ must be as well, since $u$ is a duplicate of $v$ with possibly some more edges. Therefore if $v(i_2)$ is in $K^*$ then so is $v(i_1)$.
\end{proof}

\begin{lemma}
	A solution of Algorithm \ref{algorithm:approximation algorithm mspec} is a feasible solution for {\it MSPEC} and the sum of the assigned powers is $\alpha |K^*|$.
	\label{lemma:feasible}
\end{lemma}

\begin{proof}
    Let $G_p$ be the result of removing from $G$ all the edges $(u,v)$ with $p_u + p_v \ge w_{u,v}$ where $p$ is as defined by  Algorithm \ref{algorithm:approximation algorithm mspec}. (Recall that we always set $p_s = p_t = 0$.)
    We must show that $s$ and $t$ are disconnected in $G_p$.  Suppose not.
	Then there is an $s$-$t$ path $P$ in $G_p$, say $s, x_1, x_2, \ldots, x_l, t$.
	We will show that the path $P' = s, x_1(k_{x_1}^*), x_2(k_{x_2}^*), \ldots, x_l(k_{x_l}^*), t$ exists in $G' - K^*$, a contradiction to $K^*$ being an $s$-$t$ vertex cut.
	
	First, note that all the vertices of $P'$ lie in $G' - K^*$ by Claim~\ref{claim:nestedstructure}.
	Consider an edge $(u,v)$ of the path $P$ with $u,v \in V$.  
	Since this edge is still present in $G_p$, we know that $p_u + p_v < w_{u,v}$.  
	Thus by the definition of the $p$ values,  $k^*_u \alpha + k^*_v \alpha < w_{u,v}$. 
	By the construction of $G^\prime$, the edge $(u(k_u^*), v(k_v^*))$ is in $E^\prime$.
	
	It remains to consider the edges of $P$ incident to $s$ and $t$.  Since the edge $(s,x_1)$ is in $G_p$, we know that
	$p_{x_1} < w_{s,x_1}$.  Thus by the definition of the $p$ values, $k^*_{x_1} \alpha < w(s, x_1)$.  By the construction of $G'$, the edge $(s, x_1(k_{x_1}^*))$ is in $E'$.
	The argument for edge $(x_l, t)$ is similar.  
	
	Thus the path $P'$ exists in $G' - K^*$ which gives the desired contradiction. 
	
	In our solution we assign a power of $p_v$ to vertex $v$, so the total power assigned is:
	\begin{ceqn}
	    \begin{equation}
	        \sum_{v \in V} p_v = \sum_{v \in V} k^*_v \cdot \alpha = \alpha \sum_{v \in V} k^*_v = \alpha |K^*|
	    \qedhere\end{equation}
	\end{ceqn}

\end{proof}

Next we prove an upper bound on the size of the set $K^*$ (a minimum $s$-$t$ vertex cut in $G'$) in terms of an optimum solution to \MSPEC. 
Let $p^*_v, v \in V$ be an optimum power assignment for \MSPEC. Let $p^*_s = p^*_t = 0$.

\begin{lemma}
    $G'$ has a vertex cut $K$ of size $\sum \lceil \frac{p_v^*}{\alpha} \rceil$.  Thus $|K^*| \le \sum \lceil \frac{p_v^*}{\alpha} \rceil$.
\label{lemma:VC-bound}
\end{lemma}

\begin{proof}
Let $M^*$ be the set of edges of $G$ that are removed by the optimum power assignment $p^*_v, v \in V$, i.e., $M^* = \{(u,v) : p^*_u + p^*_v \ge w_{u,v} \}$. 
$M^*$ is an $s$-$t$ cut in $G$.

Define $k_v = \lceil \frac{p_v^*}{\alpha} \rceil$ for $v \in V$.
Define $K$, a set of vertices of $G'$, to consist of the first $k_v$ copies of each vertex $v \in V$, 
i.e.,~$K = \{v(i): v \in V, 0 \le i < k_v \}$.  Observe that $|K| = \sum \lceil \frac{p_v^*}{\alpha} \rceil$.  
To prove the lemma, we just need to show that $K$ is an $s$-$t$ vertex cut in $G'$.

It suffices to show that if $(u,v) \in M^*$ then there is no copy of edge $(u,v)$ in $G' - K$.
Consider an edge $(u,v) \in M^*$.  Then $p^*_u + p^*_v \ge w_{u,v}$, so $\lceil \frac{p_u^*}{\alpha} \rceil \alpha + \lceil \frac{p_v^*}{\alpha} \rceil \alpha \ge w_{u,v}$, i.e., $k_u \alpha + k_v \alpha \ge w_{u,v}$.   

Now observe that if $i \ge k_u$ and $j \ge k_v$ then $i \alpha + j \alpha \ge k_u \alpha + k_v \alpha \ge w_{u,v}$, so $(u(i), v(j))$ is not an edge of $G'$ (by definition of $G'$).  On the other hand, if $i < k_u$ then $u(i) \in K$, so $u(i)$ is not a vertex of $G' - K$, and similarly, if $j < k_v$ then $v(j) \in K$, so $v(j)$ is not a vertex of $G' - K$. 
Thus no copy of the edge $(u,v)$ exists in $G' - K$.  This proves that $K$ is an $s$-$t$ vertex cut in $G'$ and completes the proof of the lemma.
\end{proof}

\begin{theorem}
	Algorithm~\ref{algorithm:approximation algorithm mspec} is a fully polynomial time approximation scheme (FPTAS) for the Minimum Shared-Power Edge Cut problem (MSPEC).
	Furthermore, the running time of Algorithm~\ref{algorithm:approximation algorithm mspec} is $O(n^{5.5}m\varepsilon^{-2.5})$.
\end{theorem}

\begin{proof}
	 
	 By Lemma~\ref{lemma:feasible} we know that Algorithm~\ref{algorithm:approximation algorithm mspec} gives a solution to \MSPEC of cost $\alpha |K^*|$, 
	 and by Lemma~\ref{lemma:VC-bound} we have $|K^*| \le \sum \lceil \frac{p_v^*}{\alpha} \rceil$.

\begin{ceqn}
	 \begin{equation*}
    	 \begin{aligned}
    	 \text{Approximate solution} ~& =~ \alpha |K^*|  \le \sum_{v \in V} \Big \lceil \frac{p_v^*}{\alpha} \Big \rceil \cdot \alpha  
    	 ~<~ \sum_{v \in V} \Big(\frac{p_v^*}{\alpha}+1 \Big) \cdot \alpha \\ 
    	 &=~ \sum_{v \in V}  \Big (p_v^* + \alpha \Big) 
    	 ~=~ \sum_{v \in V} p_v^* + n\alpha 
    	 ~=~ OPT + n\alpha .
    	 \end{aligned}
	 \end{equation*}	
\end{ceqn}
	 
Thus the approximation ratio is at most $ \frac{OPT+n\alpha}{OPT} = \frac{OPT + \frac{n \varepsilon p^*}{n}}{OPT}  = 1 + \frac{\varepsilon p^*}{OPT} \leq (1 + \varepsilon)$, since $ OPT \geq p^*$. 
 The running time of Algorithm~\ref{algorithm:approximation algorithm mspec} depends on the running time of finding a minimum $s$-$t$ vertex cut and this can be done in time $O(n^{1/2}m) = O(n^{2.5})$ for a graph with $n$ vertices and $m$ edges using a modified version of Dinitz's flow algorithm \cite[Chapter 1]{Nagamochi:2008:AAG:1434866}~\cite[Corollary 9.7a]{schrijver2002combinatorial}. 
 If our original graph $G$ has $n$ vertices and $m$ edges, then the graph $G'$ has $|V'|=nc=O(\frac{n^3}{\varepsilon})$ vertices and at most $|E'| \le |E|c^2= O(\frac{mn^4}{\varepsilon^2})$ edges,
 and Algorithm~\ref{algorithm:approximation algorithm mspec} finds a vertex cut in graph $G'$ in $O(n^{5.5}m\varepsilon^{-2.5})$ time.
\end{proof}

\section{Variations of \MSPEC}
\label{sec:variations}

In this section we give polynomial time algorithms for three special cases of \MSPEC,  
when the edge weights are uniform, or integral, and when the power values at each vertex are restricted to a polynomially bounded domain.
In the final section we show that our FPTAS extends to a more general version of \MSPEC where there are costs on the vertices.

\subsection{Uniform/Integral edge weights}
\label{subsec:uniform}

We use the alternative formulation of \MSPEC given at the start of Section~\ref{sec:formulations}. 
For any instance of \MSPEC there is a vertex partition $S,T$ with $s \in S$ and $t \in T$, such that the minimum power sum is 
equal to the minimum ``$w$-vertex cover'' in the bipartite graph $B$ of edges between $S$ and $T$, 
and this is equal to the weight of a maximum matching in $B$.  
By K\H{o}nig's Theorem and its generalization to weights (Egerv\'ary's Theorem, see~\cite[Theorem 17.1]{schrijver2002combinatorial})
we know that if the edge weights are integral then there is an optimal solution where the power values are also integral, and if the edge weights are all 1 then it suffices to consider power values that are $\{0,1\}$-valued.

For \MSPEC with uniform edge weights, we can scale so that all edge weights are 1.  Then, by the above, the power values will be 0 or 1 and the problem reduces to Minimum Vertex Cut, which can be solved in time $O(n^{1/2}m)$ time as discussed in the previous section. 

For \MSPEC with integer edge weights, we can assume that the power values are integral.
Furthermore, if the edge weights are bounded by $W$, then so are the power values.
We can then use the same approach as Algorithm \ref{algorithm:construction}, but make $W+1$ copies of each vertex, and set $\alpha=1$. 
This solves \MSPEC exactly, with a running time of $O((nW)^{1/2}(mW^2))=O(n^{1/2}mW^{5/2})$.  Thus it provides a pseudo polynomial time algorithm for \MSPEC with integral edge weights.

\subsection{Polynomial Domain}
\label{subsec:polynomial domain}

In activation network design problems, one of the most common assumptions is that 
of a `polynomial domain' for the power values~\cite{Panigrahi:2011:SND:2133036.2133114,nutov2013survivable}. 
This means that the power values for vertex $v$ come from a set $D^v$ whose size is bounded by a polynomial in $n = |V|$.
The polynomial domain assumption is realistic for wireless networks when there are a small number of possible powers that can be assigned to a vertex. 
However, this does not apply to the minimum shrinkage problem 
because we are using shrinkage as a measure of barrier coverage rather than making any assumption about sensor powers.

Under the polynomial domain assumption, \MSPEC admits a polynomial time algorithm.  We sketch the algorithm.
The technique is similar to the one we used in the previous section.  
We assume $0 \in D^v$ for all $v$ (otherwise, we must pay $\min(D^v)$ directly, and can then subtract it from all later values).
Let the values in $D^v$ be $d^v_0, d^v_1, \ldots , d^v_{c_v}$ in increasing order.
We modify the construction of graph $G'$ from Algorithm~\ref{algorithm:construction}, to create $|D^v| = c_v +1$ copies of each vertex $v$, which we number $v(0), v(1), \ldots, v(c_v)$. 
We assign a cost of $d^v_{i+1} - d^v_i$ to $v(i)$ for $i=0, \ldots, c_v -1$, and $\infty$ to $v(c_v)$.

In graph $G'$ we add an edge between two vertices $u(i)$ and $v(j)$ if $d^u_i + d^v_j < w_{u,v}$.
We now run a max flow algorithm to find the minimum cost $s$-$t$ vertex cut in $G'$  
by applying the standard graph transformation (see~\cite{west}) to convert a vertex cut to an edge cut, and vertex weights to edge capacities.
If the minimum cost is $\infty$ then there is no feasible assignment of powers from the given domains.  Otherwise, 
for vertex $v$, if $i$ is the maximum index for which $v(i)$ is included in the cut, then we assign $p_v = d^v_{i+1}$.  Note that a feasible solution has $i < c_v$ so this is well-defined.

For correctness, first observe that the prefix property, as mentioned in Claim \ref{claim:nestedstructure}, still holds.
Next, we claim that if the assigned powers do not remove edge $(u,v)$ then there was a copy of that edge in $G'$ after removing the optimum vertex cut.
Suppose we assigned $p_u = d^u_{i+1}$ and $p_v = d^v_{j+1}$ but $p_u + p_v < w_{u,v}$. 
Then the edge $(u(i+1), v(j+1))$ still exists in $G'$.
In the other direction, a solution to \MSPEC with powers from the given domains yields a vertex cut in $G'$ of equal cost. 
 
The running time for the algorithm is determined by the running time of the  standard max-flow algorithms, which is $O(n^3 / \log n)$~\cite{cheriyan1996n} (or see~\cite[Section 10.8]{schrijver2002combinatorial}). The number of vertices in the new graph $G^\prime$ is $\sum_v |D^v|$, thus the running time is $O((\sum_v |D^v|)^{3} / \log(\sum_v |D^v|))$.

\subsection{\MSPEC with Vertex Costs}

In this section we outline an FPTAS for a more general version of \MSPEC with multiplicative vertex costs.  Formally, the problem is defined as:

\smallskip
{\bf Input:} Graph $G=(V \cup \{s,t\},E)$ with a non-negative weight $w_{u,v}$ on each edge $(u,v) \in E$ and a non-negative cost $c_v$ on every vertex $v \in V$ .

{\bf Problem:} Assign a non-negative power $p_v$ to each $v \in V$ and assign $p_s=p_t=0$ such that removing the edge set $\{(u,v) \in E : p_u + p_v \ge w_{u,v} \}$ 
disconnects 
$s$ and $t$ and $\sum_{v \in V} c_v p_v$ is minimized.

Using a similar bottleneck and discretization approach, we can get an FPTAS for 
\MSPEC with vertex costs.

The key idea is to consider the power-cost $c_v  p_v$ instead of the power at each vertex. 
This enables us to get a bottleneck solution analogous to Theorem \ref{theorem:bottleneck} by doing a binary search on the vertex power-costs. 
Similarly, for the construction phase of graph $G^\prime$ as described in Algorithm \ref{algorithm:construction}, 
we add an edge between $u(i)$ and $v(j)$ to $G^\prime$ if $\frac{i\alpha}{c_u} + \frac{j\alpha}{c_v} < w_{u,v}$. 
The remainder of the algorithm and the proofs carry over after similar modifications.

\section{Faster Approximation for \MSPEC}
\label{sec:speedup}

In this section, we will briefly mention how to speed up our FPTAS from the current running time of $O(n^{5.5}m\varepsilon^{-2.5})$ to $O(n^{3}m\varepsilon^{-2.5} + m^3 / \log m)$.
To accomplish this, we discretize the problem using
a faster $2$-approximation for \MSPEC instead of 
the $n$-approximation given by the minimum bottleneck shared power edge cut.
Let \OPT be the optimum solution value for \MSPEC.  

For the 2-approximation, we consider a discrete version of
\MSPEC where
a vertex can only be assigned a power equal to the weight of one of its incident edges.

With this restriction, the number of possible choices of power assignment at a given vertex is equal to the degree of the vertex in $G$.

Thus the polynomial domain assumption is satisfied and the algorithm of Subsection \ref{subsec:polynomial domain} solves this discrete version in time $O((\sum_v {\rm deg} (v))^{3} / \log (\sum_v {\rm deg} (v))) = O(m^{3} / \log m)$.

We now use the fact that the discrete \MSPEC provides a 2-approximation.  This was proved more generally by Panigrahi~\cite[Lemma 3.11]{Panigrahi:2011:SND:2133036.2133114}, but we include his short argument.

\begin{theorem}\cite{Panigrahi:2011:SND:2133036.2133114}
     Discrete \MSPEC is an 2-approximation 
     for \MSPEC.
     \label{theorem:2-approx}
\end{theorem}

\begin{proof}
We will show that $2 \OPT$ is an upper bound for the discrete version of \MSPEC, which proves the theorem.
Let $C$ be the set of edges in the $s$-$t$ cut determined by an optimum solution. For each edge $e =(u,v) \in C$, select the endpoint $v$ of higher power.  Then $p_v \ge \frac{1}{2} w_{u,v}$. For each selected vertex we raise $p_v$ to the maximum $w_{x,v}$ where $(x,v)$ is an edge of $C$ that selects $v$.  Then we at most double the power of $v$. 
Furthermore, any vertex not selected can have its power decreased to 0.  The new power values still activate all edges of $C$ and provide a solution to discrete \MSPEC.
\end{proof}

Let $Z$ be the optimal value of a  solution to the discrete \MSPEC problem on graph $G$.  Clearly $\OPT \leq Z$, and from Theorem \ref{theorem:2-approx}, we have $Z \leq ~ 2 \OPT$. Thus:
\begin{ceqn}
\begin{equation}
     \frac{Z}{2} \leq OPT \leq Z
    \label{eqn:OPT 2-approx}
\end{equation}
\end{ceqn}
We can use these bounds in place of those in Lemma~\ref{lemma:n-approx} to obtain a faster approximation algorithm for \MSPEC.
From Equation~(\ref{eqn:OPT 2-approx}), the maximum power we assign to any vertex is $Z$. If we introduce an error of at most $\alpha = \frac{\varepsilon Z}{2 n}$ at each vertex, the total error will be at most $n \alpha \leq \varepsilon \OPT$.  
We construct the graph $G^\prime$ as in Algorithm \ref{algorithm:construction}, though instead of taking $\lceil \frac{n^2}{\varepsilon} \rceil$ copies of each vertex, we only use $c= \lceil \frac{Z}{\alpha} \rceil = \lceil \frac{2Zn}{\varepsilon Z} \rceil =\lceil \frac{2n}{\varepsilon} \rceil $ copies. 

Algorithm~\ref{algorithm:approximation algorithm mspec} on this new $G^\prime$ is an FPTAS with a running time of $O(|V'|^{1/2}|E'|) $= $O((n\frac{n}{\varepsilon})^{1/2}(m(\frac{n}{\varepsilon})^2))$ $= O(n^3 m \varepsilon^{-2.5})$. Since the time to compute $Z$ is $O(m^3 / \log(m))$, the overall running time of our fast FPTAS is $O(n^3m\varepsilon^{-2.5} + m^3 / \log m)$.
In addition to improving the running time, this algorithm also improves the space complexity by at least $O(n^2)$.


\section{Open Problems}
\label{sec:conclusions}

The biggest open question is to settle the complexity of the \MSPEC problem---is it in P?  NP-hard?
For hardness, a starting point would be to show that \MSPEC with vertex costs is NP-hard.
On the other side, the complexity of the minimum shrinkage problem (i.e.,~the special case of \MSPEC arising from shrinking unit disks in the plane) is also open. 
Can the geometry of unit disk graphs be used to develop a polynomial time algorithm?

\medskip
\noindent\textbf{Acknowledgements.} This work was begun at the Fifth Annual Workshop on Geometry and Graphs, held at the Bellairs Research Institute in Barbados, March 5-10, 2017. We thank the organizers and all participants for the productive and positive atmosphere.
We thank Lap Chi Lau for helpful suggestions.


\bibliographystyle{plain}
\bibliography{biblio}

\end{document}